\documentclass[conference]{IEEEtran}
\usepackage{amsmath,amsthm,paralist,color}
\newcommand{\C}{\mathbf{C}}

\newcommand{\ket}[1]{|{#1}\rangle}
\newcommand{\nix}[1]{}
\DeclareMathOperator{\Tr}{tr}
\DeclareMathOperator{\trace}{Tr}
\DeclareMathOperator{\wt}{wt}

\newtheorem{proposition}{Proposition}

\newtheorem{lemma}[proposition]{Lemma}

\begin{document}
\title{Hybrid Codes}
\author{\IEEEauthorblockN{Andrew Nemec}
\IEEEauthorblockA{Texas A\&M University\\
Department of Computer Science and Engineering\\ 
College Station, TX 77845-3112\\
Email: nemeca@cse.tamu.edu}
\and
\IEEEauthorblockN{Andreas Klappenecker}
\IEEEauthorblockA{Texas A\&M University\\
Department of Computer Science and Engineering\\ 
College Station, TX 77845-3112\\
Email: klappi@cse.tamu.edu}}
\maketitle

\begin{abstract}
  A hybrid code can simultaneously encode classical and quantum
  information into quantum digits such that the information is
  protected against errors when transmitted through a quantum channel.
  It is shown that a hybrid code has the remarkable feature that it
  can detect more errors than a comparable quantum code that is able
  to encode the classical and quantum information. Weight enumerators
  are introduced for hybrid codes that allow to characterize the
  minimum distance of hybrid codes. Surprisingly, the weight
  enumerators for hybrid codes do not obey the usual MacWilliams identity.
\end{abstract}

\section{Introduction}
A hybrid code can simultaneously encode classical and quantum
information into quantum digits such that the information is protected
against errors when transmitted through a quantum channel. We will
show that hybrid codes have the remarkable feature that they can
always detect more errors than quantum error detecting codes. So
hybrid codes are in general preferable to quantum error detecting
codes for the simultaneous transmission of classical and quantum
information over a quantum channel.

In their seminal paper~\cite{devetak2005}, Devetak and Shor
characterized the set of admissible rate pairs for simultaneous
transmission of classical and quantum information over a given quantum
channel. They showed that time-sharing a quantum channel for the
separate encoding of quantum and classical information is inferior to
simultaneous transmission. This line of research was extended in
various directions. For instance, Hsieh and Wilde~\cite{hsieh10}
considered the problem of simultaneous transmission of classical and
quantum information over an entanglement-assisted quantum channel.
Yard, Hayden and Devetak \cite{yard2008} considered multi-access
channels with two senders and one receiver to communicate both
classical and quantum information to the receiver. 
There are more papers in quantum information theory
about the simultaneous transmission of classical and quantum
information, but the small selection that we have mentioned should
convey the flavor of this line of research. 

We need codes to transmit classical and quantum information over a
quantum channel. Of course, we can always use a quantum
error-correcting code for this purpose, and simply encode the
classical information in some quantum bits. However, this fails to
take advantage of gains promised by quantum information theorists.
Surprisingly, the foundations of hybrid code have not yet been well
developed. We are aware of a few notable exceptions.  Kremsky, Hsieh,
and Brun investigated early on entanglement-assisted hybrid stabilizer
codes~\cite{kremsky}.  Beny, Kempf, and Kribs briefly sketched an
operator-theoretic construction of hybrid codes~\cite{beny2007}, an
approach that has much potential.  More recently, Grassl, Lu, and Zeng
\cite{grassl} gave a number of hybrid code constructions, derived
linear programming bounds for hybrid stabilizer codes, and found very
remarkable examples of hybrid codes with good parameters.

In the next section, we define the notion of detectable errors of a
hybrid code. We show that hybrid codes can detect more errors than
comparable quantum codes. In Section~\ref{sec:weight}, we introduce
weight enumerators for hybrid codes. As in the case of quantum codes,
we have two weight enumerators. For one of the weight enumerators, we
use the average of the Shor-Laflamme weight enumerators for the
quantum codes that encode in the quantum information. We show that the
two weight enumerators allow us to characterize the errors that can be
detected and corrected by the hybrid code. In
Section~\ref{sec:macwilliams}, we show the unexpected result that
weight enumerators of a hybrid code do not satisfy the MacWilliams
identity, but rather a relaxed version of the MacWilliams identity.

\section{Hybrid Codes}\label{sec:hybrid}
Suppose that we want to simultaneously transmit classical and quantum
messages. Our goal will be to encode them into the state of $n$
quantum digits that have $q$-levels each, so that the encoded message
can be transmitted over a quantum channel. In other words, an encoded
message is a unit vector in the Hilbert space
$$ H = \bigotimes_{k=1}^n \C^{q} \cong \C^{q^n}.$$

A hybrid code has the parameters $((n,K\!:\!M))_q$ if and only if it
can simultaneously encode one of $M$ different classical messages and
a superposition of $K$ orthogonal quantum states into $n$ quantum
digits with $q$ levels. We can understand the hybrid code as a
collection of $M$ orthogonal $K$-dimensional quantum codes $C_m$ that
are indexed by the classical messages
$m\in [M] := \{1, 2, \ldots, M\}$.
If we want to transmit a classical message $m\in [M]$ and a quantum
state $\varphi$, then we need to encode $\varphi$ into the quantum
code $C_m$. 

The encoded states will be subject to errors when transmitted through
a quantum channel.  Our first task will be to characterize the errors
that can be detected by the hybrid code. We will set up a projective
measurement that either upon receipt of a state $\ket{\psi}$ in $H$
either (a) returns $\epsilon$ to indicate that an error happened or
(b) or claims that there is no error and returns a classical message
$m$ and a projection of $\ket{\psi}$ onto $C_m$.

Let $P_m$ denote the orthogonal projector onto the quantum code
$C_m$ for all integers $m$ in the range $1\le m\le M$.  For distinct
integers $a$ and $b$ in the range $1\le a,b\le M$, the quantum codes
$C_a$ and $C_b$ are orthogonal, so $P_bP_a=0$. It follows that the orthogonal
projector onto $C =\bigoplus_{m=1}^M C_m$ is given by 
$$P=P_1+P_2+\cdots+P_M.$$ 
We define the orthogonal 
projection onto $C^\perp$ by $P_\epsilon = 1-P$. 

For the hybrid code $\{ C_m \mid m\in [M]\}$, we can define a
projective measurement $\mathcal{P}$ that corresponds to the set
$$ \{ P_1, P_2, \ldots, P_M, P_\epsilon\}$$ of projection operators
that partition unity.

We can now define the concept of a detectable error.  An error $E$ is
called \textit{detectable}\/ by the hybrid code
$\{ C_m \mid m\in [M]\}$ if and only if for each index $a, b$ in the
range $1\le a, b \le M$, we have
$$ 
P_b E P_a = \begin{cases}
\lambda_{E,a} P_a & \text{if $a=b$}, \\
0 & \text{if $a\neq b$}
\end{cases}
$$
for some scalar $\lambda_{E,a}$. 

The motivation for calling an error $E$ detectable is the following
simple protocol. Suppose that we encode a classical message $m$ and a
quantum state into a state $v_m$ of $C_m$, and transmit it
through a quantum channel that imparts the error $E$. If the error is
detectable, then measurement
of the state $Ev_m = EP_mv_m$
with the projective measurement $\mathcal{P}$  either 
\begin{compactenum}[(E1)]
\item returns $\epsilon$, which signals that an error happened, or 
\item returns $m$ and corrects the error by projecting the state back
  onto a scalar multiple $\lambda_{E,m} v_m = P_mEP_mv_m$ of the state
  $v_m$.
\end{compactenum}
The definition of a detectable error ensures that the measurement
$\mathcal{P}$ will never return an incorrect classical message $d$,
since $P_d E P_m v_m= 0$ for all $d\neq m$, so the probability of
detecting an incorrect message is zero. An error that is not
detectable by the hybrid code can change the encoded classical information, the
encoded quantum information, or both.

The next proposition shows that hybrid codes can always detect more
errors than a comparable quantum code that encodes both classical and
quantum information. This is remarkable given that the advantages are
much less apparent when one considers minimum distance,
see~\cite{grassl}. 

Let $B(H)$ denote the set of linear operators on $H$.

\begin{proposition}
The subset $\mathcal{D}$ of detectable errors in $B(H)$ of an $((n,
K\!:\!M))_q$  hybrid
code form a vector space of dimension
$$\dim \mathcal{D} = q^{2n} - (MK)^2 + M.$$
In particular, an $((n, K\!\mathop{:}\!M))_q$ hybrid code with $M>1$
can detect more errors than an $((n, KM))_q$ quantum code.
\end{proposition}
\begin{proof}
It is clear that any linear combination of detectable errors is
detectable. If we choose a basis adapted to the orthogonal decomposition
$H= C \oplus C^\perp$ with 
$$C=C_1\oplus C_2\oplus \cdots \oplus C_M,$$
then an error $E$ is represented by a matrix of the form 
$$ 
\left(\begin{array}{cc}
A & R \\ 
S & T
\end{array}\right)
$$
Since $E$ is detectable, the $MK\times MK$ matrix $A$ must satisfy 
$$A = \lambda_{E,1} 1_K \oplus \lambda_{E,2} 1_K \oplus \cdots \oplus
\lambda_{E,M} 1_K,$$
where $1_K$ denote a $K\times K$ identity matrix, but $R$, $S$, and
$T$ can be arbitrary. Therefore, the dimension of the vector space of
detectable errors is given by $q^{2n} - (MK)^2 + M$. The vector space
of detectable errors of an $((n, KM))_q$ quantum code has dimension 
$q^{2n}-(KM)^2+1$, which is strictly less than $q^{2n} - (MK)^2 + M$.
\end{proof}

\newcommand{\one}{\mathbf{1}}
We conclude this section with a few remarks on sets of detectable and
correctable errors. Detectable errors have many nice features. The set
$\mathcal{D}$ of all detectable errors of a hybrid code is a vector
space that contains the identity operator, is closed under taking
adjoints $*$, and is a closed subspace of $B(H)$. Therefore, the set
$\mathcal{D}$ of detectable errors is an operator system of the the
$C^*$-algebra $B(H)$. This means that we can express every detectable
error in $\mathcal{D}$ as a linear combination of detectable errors
that are positive operators. Indeed, an operator $E$ in $\mathcal{D}$
can be expressed as linear combination $E=A+iB$, where
$A=\frac{1}{2}(E+E^*)$ and $B=\frac{i}{2}(E^*-E)$ are self-adjoint
operators in $\mathcal{D}$. A self-adjoint operator $X$ in
$\mathcal{D}$ can be expressed as the difference of the positive
operators $\|X\|\one$ and $\|X\|\one-X$. In short, the set of
detectable errors of a hybrid code has a quite well-behaved structure. 

On the other hand, whenever we consider the correctability of errors,
we must consider an entire set of errors rather than a single
error. Depending on the set of errors that we would like to correct,
a given error operator $E$ might or might not be correctable. 
It is not difficult to show that a unital set $\mathcal{E}$ of errors
is correctable if and only if the set 
$$\mathcal{E}^*\mathcal{E} = \{ F^*E \mid E, F \in \mathcal{E}\}$$
of errors is detectable. In other words, all errors $E, F\in
\mathcal{E}$ must satisfy 
$$ P_b F^*EP_a = \lambda_{F^*E,a}\, [a=b]\, P_a$$
for all $a, b \in [M]$, where $[a=b]$ denotes the Iverson-Knuth
bracket that is equal to 1 when the condition $a=b$ is satisfied and
$0$ otherwise. 

In the next section, we will introduce the notion of a weight of
errors and introduce weight enumerators of hybrid codes.

\section{Weight Enumerators}\label{sec:weight}
In this section, we define weight enumerators for an $((n,
K\mathop{:} M)_q$  hybrid code 
$$ \mathcal{H} = \{ C_m \mid m\in [M]\}.$$
Before we can define the weight enumerators, we will briefly recall
the concept of a nice error basis (see~\cite{knill96a,klappenecker038,klappenecker034} for further details), so that we can define a suitable
notion of weight for the errors. 

Let ${G}$ be a group of order $q^2$ with identity element~1.  A
\textit{nice error basis}\/ on $\C^q$ is a set ${\cal E}=\{
\rho(g)\in {\cal U}(q) \,|\, g\in {G}\}$ of unitary matrices such that
\begin{tabbing}
i)\= (iiiii) \= \kill
\>(i)   \> $\rho(1)$ is the identity matrix,\\[1ex]
\>(ii)  \> $\trace\rho(g)=0$ for all $g\in G\setminus \{1\}$,\\[1ex]
\>(iii) \> $\rho(g)\rho(h)=\omega(g,h)\,\rho(gh)$ for all $g,h\in
{G}$,
\end{tabbing}
where $\omega(g,h)$ is a nonzero complex number depending on $(g,h)\in
G\times G$;  the function
$\omega\colon G\times G\rightarrow \C^\times$ is called the factor
system of $\rho$.
We call $G$ the \textit{index group}\/ of the error basis
${\cal E}$.  The nice error basis that we have introduced so far
generalizes the Pauli basis to systems with $q\ge 2$ levels. 

We can obtain a nice error basis $\mathcal{E}_n$ on $H\cong \C^{q^n}$ by
tensoring $n$ elements of $\mathcal{E}$, so 
$$ \mathcal{E}_n = \mathcal{E}^{\otimes n} = \{ E_1 \otimes E_2\otimes
\cdots \otimes E_n \mid E_k \in \mathcal{E}, 1\le k\le n\}.$$
The weight of an element in $\mathcal{E}_n$ are the number of
non-identity tensor components. We write $\wt(E)=d$ to denote that the 
element $E$ in $\mathcal{E}_n$ has weight $d$. 

We can associate with a hybrid code $\mathcal{H}$ two weight
enumerators 
$$A\!\left(z\right)=\sum_{d=0}^{n}A_{d}z^{d}\text{ and
}B\!\left(z\right)=\sum_{d=0}^{n}B_{d}z^{d},$$
where the coefficients are given by 
$$A_{d}=\frac{1}{K^{2}M}\sum_{a,b=1}^M\sum_{\substack{E\in \mathcal{E}_n\\ \wt(E)=d}}
\left|\Tr\!\left(P_b EP_a\right)\right|^2$$
and
$$B_d=\frac{1}{K^2M}\sum_{a,b=1}^{M}\sum_{\substack{E\in \mathcal{E}_n\\ \wt(E)=d}}
\Tr\!\left((P_bEP_a) (P_bEP_a)^*\right) \Tr(P_a).$$
We note that both sums can be considerably simplified, but we leave
them in the current form for now, since that simplifies the proof of
the next proposition. We call $(A_0, A_1, \ldots, A_n)$ and
$(B_0, B_1, \ldots, B_n)$ the weight distributions of the hybrid code
$\mathcal{H}$.

There is only one element in $\mathcal{E}_n$ of weight 0, namely the
identity matrix. The normalization constants are chosen such that $A_0=B_0=1$. 

\begin{proposition}\label{p:weights}
  Let $\mathcal{H}$ be a $((n, K\mathop{:} M))_q$ hybrid code with
  weight distributions $A_{d}$ and $B_{d}$. Then the weight
  distributions satisfy the following properties. 
\begin{compactenum}[(a)] 
\item The inequality
  $B_{d}\geq A_{d}\geq0$ holds for all integers $d$ in the range
  $0\le d\le n$.  
\item We have $A_d=B_d$ if and only if $\mathcal{H}$ can
  detect all errors in $\mathcal{E}_n$ of weight $d$. 
\end{compactenum}
\end{proposition}
\begin{proof}
\begin{compactenum}[(a)]
\item Recall that the Cauchy-Schwarz
  inequality for operators $A, B\in B(H)$ is given by 
 \begin{equation}\label{eq:cauchy}
\left|\Tr\!\left(A^*B\right)\right|^2\leq\Tr\!\left(A^*A\right)\Tr\!\left(B^*B\right)
\end{equation}
and equality holds precisely when $A$ and $B$ are linearly dependent. 

If we apply this inequality to the term $|\Tr(P_bEP_a)|^2$ in $A_d$,
then we find that 
\begin{align*}
\left|\Tr\!\left(P_bEP_a\right)\right|^2&
    =\left|\Tr\!\left( (P_bEP_a) P_a\right)\right|^2\\ 
    & \le \Tr\!\left((P_bEP_a)(P_bEP_a)^*\right)
        \Tr\!(P_a^*P_a)\\
    & =  \Tr\!\left((P_bEP_a)(P_bEP_a)^*\right)
        \Tr\!(P_a)
\end{align*}
Summing over all $a,b\in [M]$ and all error operators $E$ of weight $d$
and normalizing, we obtain $B_d\ge A_d \ge 0$.

\item If $\mathcal{H}$ can detect all errors of weight $d$ in
$\mathcal{E}_n$, then 
$$A_d = \frac{1}{M} \sum_{a=1}^M \sum_{\substack{ E\in \mathcal{E}_n
  \\ \wt(E)=d}}
|\lambda_{E,a}|^2 = B_d.$$

Conversely, if equality $A_d=B_d$ holds, then it follows that for all $a,b \in
[M]$ and every error $E$ in $\mathcal{E}_n$ of weight $d$ the Cauchy-Schwarz
inequality 
\begin{multline}   
\left|\Tr\!\left( (P_bEP_a) P_a\right)\right|^2\\ 
\le  \Tr\!\left((P_bEP_a)(P_bEP_a)^*\right) \Tr\!(P_a^*P_a)
\end{multline}
holds with equality. Therefore, $P_bEP_a$ and
$P_a$ are linearly dependent for all $a, b\in [M]$ and all $E$ with
$\wt(E)=d$. We will distinguish between
\begin{inparaenum}[(i)] \item the diagonal case $a=b$
and \item the off-diagonal case $a\neq b$. 
\end{inparaenum}

\begin{compactenum}[(i)]
\item If $a=b$, then we can deduce that for each $a\in
[M]$ and each error operator $E$ of weight $d$ there exists a scalar
$\lambda_{E,a}$ such that 
$$ P_bEP_a = \lambda_{E,a} P_a.$$

\item If $a\neq b$, then both sides of the inequality are equal to 0, since 
the left-hand side satisfies 
$$\left|\Tr\!\left( (P_bEP_a)\right)\right|^2 =
\left|\Tr\!\left( P_bEP_a P_b\right)\right|^2 = 0.$$ 
On the right-hand side, we have $\Tr(P_a)=K\neq 0$, so we can deduce that
$$\Tr\!\left((P_bEP_a)(P_bEP_a)^*\right)=0.$$ 
Since $\Tr(XX^*)=\|X\|^2=0$
implies that $X=0$, we can conclude that $P_bEP_a=0$. 
\end{compactenum}
In other words, if $A_d=B_d$, then it follows from (i) and (ii) that
every error operator $E$ in $\mathcal{E}_n$ of weight $d$ is
detectable by the hybrid code $\mathcal{H}$. \qedhere
\end{compactenum}
\end{proof}

We can simplify the expressions for the coefficients $A_d$ and $B_d$
of the weight distributions of a hybrid code. The coefficients $A_d$
take a particularly simple form, namely they are equal to the average
of the Shor-Laflamme weights~\cite{shor97} of the quantum codes $C_m$
with $m\in [M]$. 
\begin{lemma}
The weight $A_d$  of an $((n,
K\mathop{:} M))_q$ hybrid code $\mathcal{H}=\{ C_m\mid m\in [M]\}$ is
obtained by averaging the Shor-Laflamme weights $A_d(C_m)$ of the
quantum codes $C_m$. In other words,
$$A_{d}=\frac{1}{K^{2}M}\sum_{a=1}^M\sum_{\substack{E\in \mathcal{E}_n\\ \wt(E)=d}}
\left|\Tr\!\left(P_a E\right)\right|^2$$
for all integers $d$ in the range $0\le d\le n$. 
\end{lemma}
\begin{proof}
The proof of the
previous proposition revealed that the off-diagonal terms in 
$$A_{d}=\frac{1}{K^{2}M}\sum_{a,b=1}^M\sum_{\substack{E\in \mathcal{E}_n\\ \wt(E)=d}}
\left|\Tr\!\left(P_b EP_a\right)\right|^2$$
vanish, since $\left|\Tr\!\left(P_b EP_a\right)\right|^2=0$ when
$a\neq b$. The diagonal terms $\left|\Tr\!\left(P_a
    EP_a\right)\right|^2$ are equal to $\left|\Tr\!\left(P_a
    E\right)\right|^2$, which proves the claim. 
\end{proof}

We can also simplify the expression 
$$B_d=\frac{1}{K^2M}\sum_{a,b=1}^{M}\sum_{\substack{E\in \mathcal{E}_n\\ \wt(E)=d}}
\Tr\!\left((P_bEP_a) (P_bEP_a)^*\right) \Tr(P_a),$$
a little bit by simplifying the argument of the first trace and noting
that $\Tr P_a = K$. Then we obtain
$$B_d=\frac{1}{KM}\sum_{a,b=1}^{M}\sum_{\substack{E\in \mathcal{E}_n\\ \wt(E)=d}}
\Tr\!\left(P_bEP_a E^*\right).$$
Unlike in the case of the weights $A_d$, the off-diagonal terms
$\Tr(P_bEP_aE^*)$ of the weight $B_d$ do not necessarily vanish. 

\section{MacWilliams Identities?}\label{sec:macwilliams}
Given that the Shor-Laflamme weights of quantum codes obey the quantum
MacWilliams identities~\cite{shor97}, it is natural to ask whether the
weight enumerators $A(z)$ and $B(z)$ of a hybrid code also satisfy the
MacWilliams identity 
$$ B(z) = \frac{K}{q^n} (1+(q^2-1)z)^nA\left(\frac{1-z}{1+(q^2-1)z}\right)?$$
Since the weight $A_d$ of an $((n, K\mathop{:} M))_q$ hybrid code is
given by the average of the $A$-weights of the quantum codes $C_m$,
it is natural to consider the average of the dual weights 
$$ A_d^\perp =
\frac{1}{KM}\sum_{a=1}^{M}\sum_{\substack{E\in \mathcal{E}_n\\ \wt(E)=d}}
\Tr\!\left(P_aEP_a E^*\right).$$
We can define the weight enumerator 
$$A^\perp(z) = \sum_{d=0}^n A_d^\perp z^d.$$
This weight enumerator captures the diagonal part $A_d^\perp$ of each weight
$B_d$.  By mimicking the proof of Shor and Laflamme~\cite{shor97} for
the MacWilliams identity for quantum codes, it
is possible to show that the average weight enumerators satisfy 
$$ A^\perp(z) = \frac{K}{q^n} (1+(q^2-1)z)^nA\left(\frac{1-z}{1+(q^2-1)z}\right).$$

If we define the off-diagonal weights 
$$ C_d = \frac{1}{KM}\sum_{\substack{a,b=1\\ a\neq b}}^{M}
\sum_{\substack{E\in \mathcal{E}_n\\ \wt(E)=d}} \Tr\!\left(P_bEP_a
  E^*\right)$$
and the corresponding weight enumerator
$$ C(z) = \sum_{d=0}^n C_d^\perp z^d,$$
then we can express the weight enumerator $B(z)$ in the form 
$$ B(z) = A^\perp(z) + C(z).$$
The coefficients of $C(z)$ satisfy $C_d\ge 0$. By Proposition~\ref{p:weights}, we have $C_d=0$ when all errors of weight $d$ are detectable by the hybrid code. 

In terms of $A(z)$, the weight enumerator $B(z)$ is given by 
$$ B(z) = \frac{K}{q^n} (1+(q^2-1)z)^nA\left(\frac{1-z}{1+(q^2-1)z}\right) + 
C(z). 
$$
Thus, the usual MacWilliams identity does not hold for hybrid codes,
but a relaxed version does.  

\section{Conclusions}\label{sec:conclusions} 
Many protocols in quantum communication require the transmission of
both classical and quantum information. Devetak and Shor showed in
\cite{devetak2005} that a time-sharing approach for the transmission
of classical and quantum information is in general inferior to a
simultaneous transmission.  The question is how to accomplish this
task. We showed that hybrid codes always offer an advantage over a
comparable quantum code, since they allow one to detect more
errors. We introduced weight enumerators for hybrid codes that allow
one to characterize the highest weight of errors that can be detected
by the code.


\end{document}